\documentclass[
 aip,
 jcp,
 amsmath,amssymb,
 reprint,
]{revtex4-2}

\usepackage{graphicx}
\usepackage[utf8]{inputenc}
\usepackage{etoolbox}
\usepackage{mathtools}
\usepackage{hyperref}
\hypersetup{hidelinks}
\usepackage{tikz}
\usetikzlibrary{arrows.meta,positioning,fit,calc}

\allowdisplaybreaks

\makeatletter
\def\@email#1#2{%
 \endgroup
 \patchcmd{\titleblock@produce}
  {\frontmatter@RRAPformat}
  {\frontmatter@RRAPformat{\produce@RRAP{*#1\href{mailto:#2}{#2}}}\frontmatter@RRAPformat}
  {}{}
}%
\makeatother

\newcommand{\Tr}{\operatorname{Tr}}
\newcommand{\Dec}{\mathcal D}

\begin{document}

\title{Reduced-Moment Representations in Electronic Structure Theory}

\author{Nan Sheng}
\affiliation{
Institute for Computational and Mathematical Engineering (ICME),
Stanford University, Stanford, CA 94305, USA.
}
\email{nansheng@stanford.edu}

\date{November 19, 2026}

\begin{abstract}
Many equilibrium electronic-structure theories use reduced quantities, such as the density, the one-body reduced density matrix (1RDM), or the one-particle Green's function, rather than the full many-electron state as their primary variables. We compare these representations in a common framework: a source generates the retained quantity, the inverse relation identifies the sources compatible with it, and further physical quantities are evaluated from the source. Static Hermitian sources, such as external potentials, admit global convex duality and define intrinsic functionals. Auxiliary sources that probe a fixed interacting system, such as imaginary-time or bilocal sources, admit local Legendre constructions on invertible source branches. We compare the static density, the 1RDM, the imaginary-time density, and the one-particle Green's function by the information they retain, the sources that generate them, and their inverse problems. Variational and stationary approximations preserve this structure through restricted state searches, approximate intrinsic functionals, or approximate effective actions. Quantum embedding is discussed as an application in which these quantities couple global and local descriptions.
\end{abstract}

\maketitle

\section{Introduction}
\label{sec:introduction}

Many equilibrium electronic-structure theories are organized around reduced
quantities rather than the full many-electron state. Density-functional
theory uses the density~\cite{HohenbergKohn1964,Mermin1965,Lieb1983}, while
reduced-density-matrix functional theories use reduced density
matrices~\cite{Coleman1963,Gilbert1975,Mazziotti2012}.
Imaginary-time formulations involve the time-dependent
density~\cite{ValievFernando1997,PolonyiSailer2002} and the one-particle
Matsubara Green's function~\cite{MartinSchwinger1959,LuttingerWard1960}.
We examine these four reduced-moment representations in terms of their
source relations and the physical information they provide.
Reduced moments also serve as interfaces between different descriptions
in embedding theories~\cite{SunChan2016QuantumEmbedding}.

Source functionals and effective actions have long been used in
density- and Green's-function theories
\cite{Fukuda1994,LuttingerWard1960,BaymKadanoff1961,
ChitraKotliar2000,ChitraKotliar2001,Kotliar2006}.
Relations between different reduced quantities are also well
established, including contractions of reduced density
matrices~\cite{Coleman1963,Gilbert1975} and the connection between
Green's-function and 1RDM functionals~\cite{Blochl2013}.
Recent work has clarified the scope and representability of static
reduced-variable functional theories~\cite{Penz2023,LiebertChaouSchilling2023,WangLiebertPenzSchilling2026}.
The purpose here is to compare the representations on common terms:
which moment is retained, which sources generate it, what the inverse relation
recovers, and what additional structure evaluates other quantities.

We use finite-temperature imaginary time as a common equilibrium setting,
with zero-temperature limits taken where they exist.  A key distinction runs
through the analysis: global convex duality applies to the admissible static
Hermitian source family, whereas global convexity is not assumed for auxiliary imaginary-time
sources about a fixed physical system.

Section~\ref{sec:orientation} distinguishes state-level and reduced-moment
strategies. Section~\ref{sec:framework} introduces the source framework,
and Secs.~\ref{sec:applications} and~\ref{sec:refinement} compare the
four full-system representations. Section~\ref{sec:embedding} discusses
variational and stationary approximations and their application to quantum
embedding. Sec.~\ref{sec:conclusion} concludes the paper.

\section{Physical orientation: state-level, reduced-moment, and embedding descriptions}
\label{sec:orientation}

Before introducing the formal source--moment language, we distinguish two
ways of representing many-electron information: state-level and
reduced-moment descriptions.  Embedding couples global and local
descriptions through selected reduced information at their interface.
Each description may use either strategy.

\begin{figure*}[tbp]
\centering
\begin{tikzpicture}[
  font=\sffamily\small,
  >=Latex,
  box/.style={draw, rounded corners=2pt, very thick, align=center,
              minimum height=1.05cm, inner sep=5pt},
  smallbox/.style={draw, rounded corners=2pt, thick, align=center,
                   minimum height=0.68cm, inner sep=4pt},
  arrow/.style={->, very thick},
  returnarrow/.style={->, thick},
  enlargearrow/.style={-{Latex[length=2.8mm,width=2mm]}, very thick, dashed},
  note/.style={align=center, font=\sffamily\scriptsize}
]

% Panel headings
\node[font=\sffamily\bfseries] at (-5.35,2.55) {(a) State level};
\node[font=\sffamily\bfseries] at (-0.15,2.55) {(b) Reduced moments};
\node[font=\sffamily\bfseries] at (5.45,2.55) {(c) Embedding interface};

% (a) State level
\node[box, fill=blue!8, minimum width=3.35cm] (stateobj) at (-5.25,1.25)
{state or state-generating object\\$\Psi,\ \Gamma,\ \Psi_\theta$, tensors, walkers};
\node[box, fill=gray!10, minimum width=2.65cm] (stateout) at (-5.25,-0.55)
{observables and tasks};
\draw[arrow] (stateobj) -- node[right,note]{evaluation\\or estimator} (stateout);
\node[note] at (-5.25,-1.75)
{The ansatz specifies the state class;\\solvers and sampling control its evaluation.};

% (b) Reduced-moment enlargement hierarchy
\node[smallbox, fill=green!10, minimum width=1.20cm] (rho) at (-1.45,-0.55)
{$\rho$};
\node[smallbox, fill=yellow!14, minimum width=1.20cm] (gamma) at (1.10,-0.55)
{$\gamma$};
\node[smallbox, fill=yellow!14, minimum width=1.20cm] (rhotau) at (-1.45,1.15)
{$\rho(\tau)$};
\node[smallbox, fill=orange!14, minimum width=1.20cm] (greenf) at (1.10,1.15)
{$G(\tau)$};
\draw[enlargearrow] (rho.east) --
node[below=7pt,note]{spatial nonlocality} (gamma.west);
\draw[enlargearrow] (rho.north) --
node[left,note]{imaginary-time\\dependence} (rhotau.south);
\draw[enlargearrow] (gamma.north) -- (greenf.south);
\draw[enlargearrow] (rhotau.east) -- (greenf.west);
\node[note, anchor=north] at (-0.15,-1.55)
{Arrows organize source--moment families.};

% (c) Embedding
\node[box, fill=blue!8, minimum width=2.15cm] (global) at (4.05,1.35)
{global\\description};
\node[box, fill=blue!8, minimum width=2.15cm] (local) at (6.45,1.35)
{local\\description};
\node[box, fill=green!10, minimum width=3.25cm] (interface) at (5.25,-0.05)
{interface moments\\$\rho_A,\ \gamma_A,\ G_{\rm loc}$};
\node[smallbox, fill=orange!12, minimum width=3.25cm] (field) at (5.25,-1.75)
{embedding variables\\$v_{\rm emb},\ u,\mathcal G_0,\Delta$};
\draw[arrow] (global.south) -- (interface.north west);
\draw[arrow] (local.south) -- (interface.north east);
\draw[arrow] (interface.south) -- (field.north);
\draw[returnarrow] (field.west) .. controls +(left:0.75cm) and +(down:0.25cm)
.. (global.south west);
\draw[returnarrow] (field.east) .. controls +(right:0.75cm) and +(down:0.25cm)
.. (local.south east);
\node[note] at (5.25,-2.55)
{Moment matching or fitting\\closes self-consistent schemes.};

% Light separators
\draw[gray!45] (-2.75,2.25) -- (-2.75,-2.85);
\draw[gray!45] (2.75,2.25) -- (2.75,-2.85);

\end{tikzpicture}
\caption{Representation-level organization of equilibrium many-electron
methods.  (a) State-level methods retain a state or a state-generating object.
(b) The static density is enlarged along two organizing directions:
equal-time spatial nonlocality leads to the 1RDM, while imaginary-time local
source dependence leads to the source-dependent density.  The one-particle Green's function
retains bilocal one-particle propagation.  The arrows compare source--moment
representations.
(c) Self-consistent embedding couples global and local descriptions, each of which may use
either strategy, through selected interface moments and feedback.}
\label{fig:structure}
\end{figure*}

\begin{table*}[tbp]
\caption{Representation strategies and their use in embedding.
Rows classify the primary represented object.}
\begin{ruledtabular}
\scriptsize
\begin{tabular}{llll}
\parbox[t]{2.6cm}{Pattern} &
\parbox[t]{3.2cm}{Primary represented object} &
\parbox[t]{5.2cm}{Additional structure used for evaluation} &
\parbox[t]{3.7cm}{Representative methods} \\
\hline
\parbox[t]{2.6cm}{State-level} &
\parbox[t]{3.2cm}{Wavefunction, density operator, or parametrized state ansatz} &
\parbox[t]{5.2cm}{Amplitude or tensor equations, optimization or projection,
sampling procedure, estimator} &
\parbox[t]{3.7cm}{HF, CI, CASSCF, CC, DMRG/MPS, neural quantum states, QMC} \\
\parbox[t]{2.6cm}{Reduced-moment} &
\parbox[t]{3.2cm}{Static or imaginary-time density, 1RDM, or one-particle Matsubara Green's function} &
\parbox[t]{5.2cm}{Constrained search or effective action, representability and
source recovery, functional or self-energy closure} &
\parbox[t]{3.7cm}{DFT, imaginary-time density-source formulations, 1RDMFT,
Green's-function theories} \\
\parbox[t]{2.6cm}{Embedding architecture} &
\parbox[t]{3.2cm}{Global and local descriptions linked by retained interface moments} &
\parbox[t]{5.2cm}{Choice of retained interface moment, global/local solvers,
feedback field, interface matching or fitting, double counting} &
\parbox[t]{3.7cm}{DFET, DMET, DMFT, \(GW\)+DMFT} \\
\end{tabular}
\end{ruledtabular}
\label{tab:structure}
\end{table*}

\subsection{State-level methods}

A state-level method represents a wavefunction, density operator, or an
implicit object from which state-level information can be generated.  Once this
object is available, observables are evaluated by expectation values, projector
equations, or estimators.  For a density operator \(\Gamma\), for example,
\begin{equation}
    \langle \hat O\rangle_\Gamma=\Tr[\Gamma\hat O].
\end{equation}
Throughout, hats distinguish field operators and observables from named
generators, such as \(H\) and \(K\), and density operators \(\Gamma\).

State-level methods include explicit wavefunctions,
parametrized state ansatzes, implicit state generators,
and stochastic sampling representations. Representative methods include Hartree--Fock~\cite{Roothaan1951}, configuration interaction~\cite{SherrillSchaefer1999},
multiconfigurational self-consistent-field methods (including CASSCF)~\cite{Siegbahn1981,Szalay2012},
and coupled cluster~\cite{Coester1958,Cizek1966}.
Tensor-network methods (including matrix-product states and
DMRG)~\cite{White1992,Schollwoeck2011,ChanSharma2011}, neural quantum states~\cite{CarleoTroyer2017,Pfau2020FermiNet}, and variational or projector Monte Carlo~\cite{Foulkes2001} also belong to this class.
Despite their different algorithms and approximations,
these methods represent or generate state-level information
from which observables are evaluated. State-level representations
also extend to finite temperature through thermofield
constructions~\cite{Harsha2019ThermofieldQC},
finite-temperature coupled cluster~\cite{WhiteChan2018FiniteTCC,
Harsha2019ThermofieldCC}, matrix-product density
operators~\cite{Verstraete2004MPDO}, and density-matrix Monte
Carlo~\cite{Blunt2014DMQMC}.

\subsection{Reduced-moment representations}

A reduced-moment representation keeps a selected physical quantity as the
primary object.  Examples include the static density~\cite{HohenbergKohn1964,Mermin1965},
a density generated by an imaginary-time local source~\cite{Matsubara1955,ValievFernando1997},
the one-body reduced density matrix~\cite{Gilbert1975}, and an
equilibrium one-particle Green's function~\cite{MartinSchwinger1959,LuttingerWard1960},
\begin{align}
\rho(\mathbf r)
&=\Tr\!\left[\Gamma\,\hat n(\mathbf r)\right],\\
\rho(\mathbf r,\tau)
&=\Tr\!\left[\Gamma\,\hat n(\mathbf r,\tau)\right],\\
\gamma_{ij}
&=\Tr\!\left[\Gamma\,\hat c_j^\dagger \hat c_i\right],\\
G_{ij}(\tau)
&=-\Tr\!\left[\Gamma\,
T_\tau \hat c_i(\tau)\hat c_j^\dagger(0)\right].
\end{align}
Here the densities are expressed in real space, while \(i,j\) label an
orthonormal one-particle basis for the matrix elements of \(\gamma\) and \(G\).
Temperature and source labels are suppressed in this overview.  For static
sources, \(\Gamma\) denotes the corresponding ground-state ensemble or
finite-temperature Gibbs state.  For imaginary-time-dependent density sources and
bilocal Green's-function sources, the trace expressions abbreviate normalized
source-weighted time-ordered insertions, as specified in
Secs.~\ref{sec:imaginary_time_density} and~\ref{sec:green_function}.

The four one-particle representations can be organized along two directions,
\begin{equation}
    \begin{array}{ccc}
        \rho(\tau) & \rightsquigarrow & G(\tau) \\
        \big\uparrow & & \big\uparrow \\
        \rho & \rightsquigarrow & \gamma .
    \end{array}
    \label{eq:two_directional_hierarchy}
\end{equation}
Along the lower horizontal direction, the 1RDM retains equal-time spatial
coherence beyond the density.  Along the vertical direction, an
imaginary-time local source generates a source-dependent density family.
At zero auxiliary source, however, the density is independent of imaginary
time, whereas the equilibrium Green's
function generally retains nontrivial imaginary-time dependence as a two-time
correlation function.

The arrows summarize changes in the retained moment, operator family, and
allowed source class.  The contractions, source restrictions, and extensions
to static \(k\)-RDMs and higher-order imaginary-time Green's functions are
discussed in Sec.~\ref{sec:refinement}.

\subsection{Embedding methods}

Embedding couples descriptions of different resolution.  A global calculation
captures the environment or extended system, while a local solver treats an
active region with a more detailed method.  Either side may use a state-level
or reduced-moment strategy.  Their communication uses selected interface
moments, such as a subsystem density, fragment 1RDM, or local Matsubara
Green's function.  In self-consistent schemes, feedback updates the effective
problems to satisfy a chosen consistency condition.

Common examples include density-functional embedding theory
(DFET)~\cite{HuangPavoneCarter2011,Libisch2014}, density-matrix embedding theory
(DMET)~\cite{KniziaChan2012}, dynamical mean-field theory
(DMFT)~\cite{Georges1996}, and \(GW\)+DMFT~\cite{Biermann2003}.
Their interface moments and feedback variables differ; the density-,
density-matrix-, and Green's-function-based formulations are compared
by Sun and Chan~\cite{SunChan2016QuantumEmbedding}.
Representation and embedding describe different aspects of a method.
DMFT, for example, uses Green's functions to couple the impurity
and lattice descriptions through self-consistency.

\section{Source--moment framework: static duality and auxiliary probes}
\label{sec:framework}

We use \(J\rightleftarrows m\) and \(J\to\Omega\) to denote
the source--moment correspondence and the evaluation of physical
quantities, respectively, where \(J\) is the source, \(m\) is the
reduced moment, and \(\Omega\) is the physical quantity of interest.  We first use \(s\) for a static
Hermitian source varying over an admissible family.  After fixing a physical
member of that family, we use \(j\) for an additional source that probes the
resulting interacting system.  These source constructions share a
representation-level organization, but have different variational properties.

\subsection{Static source families and global convex duality}

Fix \(\beta>0\) and the source-independent part \(K_0\) of the
generator. We leave the temperature dependence of the
states, moments, and responses implicit. The source \(s\)
couples to an observable family \(\hat m\) through
\begin{equation}
    K[s]=K_0+\langle s,\hat m\rangle.
    \label{eq:static_reference_family}
\end{equation}
The corresponding generating functional is
\begin{equation}
\begin{aligned}
    Z_\beta[s]
    &=\Tr[e^{-\beta K[s]}]\\
    &=\Tr\!\left[T_\tau\exp\!\left(
    -\int_0^\beta d\tau\,K[s]\right)\right],\\
    W_\beta[s]&=\beta^{-1}\log Z_\beta[s].
\end{aligned}
    \label{eq:static_generating_functional}
\end{equation}
Sources and moments are understood as elements of a dual pair
of real vector spaces, with \(\langle s,m\rangle\) denoting their pairing.
Admissible sources are chosen so that \(K[s]\) and its Gibbs
trace are well defined.
Where applicable, a fixed chemical-potential contribution
is absorbed into the one-body source.

Let
\begin{equation}
\begin{aligned}
    \Gamma[s]
    &=\frac{e^{-\beta K[s]}}{Z_\beta[s]},\\
    m[\Gamma]&=\Tr[\Gamma\hat m],\qquad
    S[\Gamma]=-\Tr[\Gamma\log\Gamma].
\end{aligned}
    \label{eq:gibbs_source}
\end{equation}
The Gibbs variational principle~\cite{Wehrl1978} gives
\begin{equation}
\begin{aligned}
    W_\beta[s]
    =\sup_{\substack{\Gamma\geq0\\\Tr[\Gamma]=1}}
    \Bigl\{&-\Tr[\Gamma K_0]
    -\langle s,m[\Gamma]\rangle
    +\beta^{-1}S[\Gamma]\Bigr\}.
\end{aligned}
    \label{eq:gibbs_variational_generator}
\end{equation}
For fixed \(\Gamma\), the variational objective is affine in \(s\), so
\(W_\beta\) is globally convex on its admissible static-source domain.
Where \(W_\beta\) is differentiable, we have
\begin{equation}
 \delta W_\beta[s]
 =-\langle\delta s,m[s]\rangle,
 \qquad m[s]=\Tr[\Gamma[s]\hat m].
 \label{eq:static_source_derivative}
\end{equation}

The static Legendre--Fenchel conjugate is
\begin{equation}
 F_\beta[m]
 =\sup_s
 \bigl\{-\langle s,m\rangle-W_\beta[s]\bigr\}.
 \label{eq:static_LF_functional}
\end{equation}
Under the standard closed-convex duality hypotheses, the source--moment
pair obeys, with \(\partial\) denoting the convex subdifferential,
\begin{equation}
 -m\in\partial W_\beta[s]
 \quad\Longleftrightarrow\quad
 -s\in\partial F_\beta[m].
 \label{eq:static_LF_duality}
\end{equation}
Thus the set of supporting sources is
\begin{equation}
 \mathcal S_\beta[m]
 =\{s:\ -m\in\partial W_\beta[s]\}.
 \label{eq:static_supporting_set}
\end{equation}
This is the static inverse-source relation.  For a regular positive-temperature
Gibbs family, two supporting sources of the same moment have coupling operators
that differ by a multiple of the identity on the admissible state space.
Thus a nonempty supporting set is a single gauge class.
With compatible state and source classes and the required
lower-semicontinuity, the conjugate agrees with the constrained search
\begin{equation}
 F_\beta[m]
 =\inf_{\Gamma\mapsto m}
 \bigl\{\Tr[\Gamma K_0]-\beta^{-1}S[\Gamma]\bigr\}.
 \label{eq:static_constrained_search}
\end{equation}
At zero temperature, where the limit exists,
\begin{equation}
 W_\infty[s]
 :=\lim_{\beta\to\infty}W_\beta[s]
 =-E[s],
 \label{eq:zero_temperature_generator}
\end{equation}
where \(E[s]\) denotes the lowest spectral value of \(K[s]\).
For a ground-state ensemble with moment \(m\), the source--moment relation
is \(-m\in\partial W_\infty[s]\); the subdifferential
accommodates degeneracies.

\subsection{Auxiliary sources about the physical system}

Choose a physical member \(s_{\rm phys}\) of the static family and fix its
complete interacting generator,
\begin{equation}
 K_{\rm phys}=K[s_{\rm phys}].
 \label{eq:physical_generator}
\end{equation}
The corresponding static moment is \(m[s_{\rm phys}]\).
We keep \(K_{\rm phys}\) fixed and introduce an auxiliary source
\(j\), which may couple to a different operator family from
\(s\).
The physical system corresponds to \(j=0\); its static argument is
suppressed in the auxiliary functional \(W_\beta[j]\), with
\(W_\beta[0]=W_\beta[s_{\rm phys}]\).

For a time-local auxiliary source \(j(\tau)\) coupled to a Hermitian
operator family \(\hat m\), the generator is given by
\begin{equation}
 K[j](\tau)=K_{\rm phys}+\langle j(\tau),\hat m\rangle.
 \label{eq:imaginary_time_generator}
\end{equation}
The corresponding generating functional is
\begin{equation}
\begin{aligned}
 Z_\beta[j]
 &=\Tr\!\left[T_\tau\exp\!\left(
   -\int_0^\beta d\tau\,K[j](\tau)\right)\right],\\
 W_\beta[j]&=\beta^{-1}\log Z_\beta[j].
\end{aligned}
 \label{eq:imaginary_time_generating_functional}
\end{equation}
For the same operator family, a time-independent auxiliary
source \(j(\tau)=s-s_{\rm phys}\) gives \(K[j](\tau)=K[s]\)
and recovers the static generating functional in
Eq.~\eqref{eq:static_generating_functional}.
Assuming a finite, nonzero source-weighted trace and the required
regularity, the source-dependent moment is
\begin{equation}
 m[j](\tau)
 =-\frac{\delta W_\beta[j]}{\delta j(\tau)}.
 \label{eq:imaginary_time_moment}
\end{equation}
It is the normalized time-ordered insertion of \(\hat m\) at imaginary
time \(\tau\); at \(j=0\), it is independent of \(\tau\) and equals
the equilibrium expectation of \(\hat m\).
We write \(m[j]\) for the full imaginary-time-dependent moment generated by \(j\).
With the normalized pairing
\begin{equation}
 \langle j,m\rangle_\beta
 =\beta^{-1}\int_0^\beta d\tau\,
    \langle j(\tau),m(\tau)\rangle,
 \label{eq:imaginary_time_pairing}
\end{equation}
the first variation of \(W_\beta\) reads
\begin{equation}
 \delta W_\beta[j]=-\langle\delta j,m[j]\rangle_\beta.
 \label{eq:auxiliary_response}
\end{equation}
The inner pairing includes the appropriate spatial integrations and spin or orbital index contractions.

The normalized partition function and its logarithm,
\begin{equation}
 \frac{Z_\beta[j]}{Z_\beta[0]},\qquad
 \log\frac{Z_\beta[j]}{Z_\beta[0]}
 =\beta\bigl(W_\beta[j]-W_\beta[0]\bigr),
 \label{eq:moment_cumulant_generators}
\end{equation}
generate full time-ordered moments and connected correlations, respectively,
through unnormalized source-kernel derivatives at \(j=0\), with a factor \((-1)^n\) for
\(n\) derivatives. 

Where the source--moment map is smooth and locally
invertible (after accounting for gauge and null directions), write
\(j=j[m]\) near \(m_{\rm phys}=m[0]\).  This defines the local
Legendre expression
\begin{equation}
 \Gamma_\beta^{\rm loc}[m]
 =-\langle j[m],m\rangle_\beta-W_\beta[j[m]],
 \label{eq:local_Legendre_expression}
\end{equation}
with
\(\delta\Gamma_\beta^{\rm loc}[m]
 =-\langle j[m],\delta m\rangle_\beta\).
Since \(j[m_{\rm phys}]=0\), the physical moment is a stationary point
of this local functional.
The second-order theory about this stationary point is discussed in the
companion study~\cite{ShengRPAHessian}.

A two-time bilocal fermionic source uses a different construction and
pairing, introduced in Sec.~\ref{sec:green_function}.

\subsection{Encoder, decoder, and reader}

The \emph{encoder} maps a source to its reduced moment; the
\emph{decoder} returns the sources compatible with that moment; and
the \emph{reader} evaluates a further target quantity.  In generic notation,
\begin{equation}
 \boxed{J\rightleftarrows m,\qquad J\longrightarrow\Omega.}
 \label{eq:three_maps}
\end{equation}
For the static theory \(J=s\); the forward map is given by
Eq.~\eqref{eq:static_source_derivative} and its inverse by
Eq.~\eqref{eq:static_supporting_set}.  For the auxiliary theory \(J=j\),
the physical point is already specified by \(j=0\), with
\begin{equation}
 m_{\rm phys}=m[0],\qquad
 \Omega_{\rm phys}=\Omega[0].
 \label{eq:physical_point}
\end{equation}
When the physical moment determines the corresponding static
moment, \(s_{\rm phys}\) can be recovered through the static
inverse relation, up to gauge freedom.
The auxiliary family probes the fixed physical system around this point.
Its decoder is
\begin{equation}
 \Dec[m]=\{j:\ m[j]=m\}.
 \label{eq:dual_decoder}
\end{equation}

\section{A hierarchy of electronic reduced-moment representations}
\label{sec:applications}

We now consider the four representations in
Eq.~\eqref{eq:two_directional_hierarchy}, identifying for each its source
class, retained moment, inverse-source relation, and task reader.
The four representations refer to the same interacting physical system
through different source classes.  Auxiliary sources probe its fixed physical
generator. Their connections are examined in
Sec.~\ref{sec:refinement}.

\subsection{Density-functional theory: the base point}
\label{sec:dft}

Ground-state DFT in Lieb's convex formulation~\cite{Lieb1983}
and finite-temperature Mermin DFT~\cite{Mermin1965}
use the density as their reduced moment and share the same source--moment
organization.  Their variational functionals are the ground-state energy and
finite-temperature free energy, respectively.
In both cases, a local one-body source \(v\) selects the system,
and the reduced moment is the density \(\rho\).  The
familiar arrows in potential coordinates are
\begin{equation}
    \boxed{
    v\rightleftarrows\rho,
    \qquad
    v\longrightarrow\Omega .
    }
    \label{eq:dft_arrows}
\end{equation}
The local-potential instance of the static construction in
Sec.~\ref{sec:framework} has \(s=v\), with
\begin{equation}
    K[v]=T+W+V[v],\qquad
    V[v]=\int d\mathbf r\,v(\mathbf r)\hat n(\mathbf r).
    \label{eq:dft_hamiltonian}
\end{equation}
Write \(\Gamma[v]=e^{-\beta K[v]}/\Tr[e^{-\beta K[v]}]\) and
\(W_\beta[v]=\beta^{-1}\log\Tr[e^{-\beta K[v]}]\).  The density is given by
\begin{equation}
    \rho[v](\mathbf r)
    =\Tr[\Gamma[v]\hat n(\mathbf r)]
    =-\frac{\delta W_\beta[v]}{\delta v(\mathbf r)}.
    \label{eq:dft_density}
\end{equation}
The inverse relation gives the supporting-potential set
\begin{equation}
    \begin{aligned}
    \mathcal V_\beta[\rho]
    &:=\{v:\ -\rho\in\partial W_\beta[v]\}\\
    &=\{v:\ -v\in\partial F_\beta[\rho]\}.
    \end{aligned}
    \label{eq:dft_supporting_potentials}
\end{equation}
At the physical potential,
\begin{equation}
    \rho_{\rm phys}:=\rho[v_{\rm phys}],\qquad
    v_{\rm phys}\in\mathcal V_\beta[\rho_{\rm phys}].
\end{equation}
At positive temperature, the equilibrium density determines
the complete local source uniquely within the full
local-potential family~\cite{Mermin1965}. With particle number fixed,
the source is instead unique up to an additive constant.  Thus, with \(T\) and \(W\) fixed, the density identifies
the equilibrium generator within this local-potential family.

The intrinsic functional is defined by the constrained search
\cite{Levy1979,Lieb1983,Mermin1965}
\begin{equation}
    F_\beta[\rho]
    =\inf_{\Gamma\to\rho}
    \left\{\Tr[\Gamma(T+W)]-\beta^{-1}S[\Gamma]\right\}.
    \label{eq:mermin_density}
\end{equation}
The search is over grand-canonical states with density \(\rho\).
Under the corresponding closed-convex duality assumptions,
\begin{equation}
    F_\beta[\rho]
    =\sup_v\{-\langle v,\rho\rangle-W_\beta[v]\}.
    \label{eq:dft_static_dual}
\end{equation}
Where the zero-temperature limit exists and the duality is preserved,
\begin{equation}
\begin{aligned}
    F_\beta[\rho]&\xrightarrow[\beta\to\infty]{}F[\rho],\\
    F[\rho]
    &=\sup_v\{E[v]-\langle v,\rho\rangle\}\\
    &=\inf_{\Gamma\to\rho}\Tr[\Gamma(T+W)].
\end{aligned}
    \label{eq:dft_zero_temperature_limit}
\end{equation}
At zero temperature, restricting the constrained search to a fixed
\(N\)-particle sector recovers Lieb's ensemble
formulation~\cite{Lieb1983}.

Kohn--Sham theory~\cite{KohnSham1965} uses the same static duality for the
noninteracting generator \(T+V[v_{\rm s}]\).
Its intrinsic free-energy functional \(F_{{\rm s},\beta}[\rho]\) follows
from Eq.~\eqref{eq:mermin_density} with \(W=0\).
Write \(F_\beta=F_{{\rm s},\beta}+E_{\rm H}+F_{{\rm xc},\beta}\),
where \(E_{\rm H}\) is the Hartree energy and \(F_{{\rm xc},\beta}\) the
exchange--correlation free energy.  Where both theories support the same
density and the functional derivatives exist, the potentials satisfy
\begin{equation}
\begin{aligned}
    v_{\rm s}&=v+v_{\rm H}+v_{\rm xc},\\
    v_{\rm H}&=\frac{\delta E_{\rm H}}{\delta\rho},\qquad
    v_{\rm xc}=\frac{\delta F_{{\rm xc},\beta}}{\delta\rho}.
\end{aligned}
    \label{eq:ks_potential}
\end{equation}
Here any additive-constant freedom is fixed.  The existence of \(v_{\rm s}\) is a noninteracting
\(v\)-representability condition~\cite{Penz2023,ShengExactDFT}. The corresponding inverse Kohn--Sham problem can be formulated through the noninteracting constrained search and its
density--potential duality~\cite{ShengInverseKS}.

The mean particle number and expectation values of local scalar
one-body operators follow directly from \(\rho\).
Other quantities, such as the kinetic energy and expectation values
of nonlocal one-body and two-body operators, can be evaluated through
appropriate density functionals or a many-body calculation at a
supporting potential \(v\). The representations considered next retain more information
than the density alone.

\subsection{Spatial enlargement: reduced-density-matrix functional theory}
\label{sec:1rdmft}

Ground-state 1RDMFT~\cite{Gilbert1975} and its finite-temperature
extensions~\cite{BaldsiefenCangiGross2015,GiesbertzRuggenthaler2019}
are generated by nonlocal one-body sources.  Their reduced moment is the 1RDM \(\gamma\), and
the same structure reads
\begin{equation}
    \boxed{
    u\rightleftarrows\gamma,
    \qquad
    u\longrightarrow\Omega .
    }
    \label{eq:1rdm_arrows}
\end{equation}
Here \(s=u\) is the complete nonlocal Hermitian one-body source.
With \(x=(\mathbf r,\sigma)\), the generator is
\begin{equation}
    K[u]=T+W+\int dx\,dx'\,
        u(x,x')\hat\psi^\dagger(x)\hat\psi(x').
\end{equation}
The physical system belongs to this source family through
its local member
\begin{equation}
    u_{\rm phys}(x,x')
    =v_{\rm phys}(\mathbf r)
       \delta(\mathbf r-\mathbf r')\delta_{\sigma\sigma'}.
\end{equation}
Write \(\Gamma[u]=e^{-\beta K[u]}/\Tr[e^{-\beta K[u]}]\).
The generated 1RDM and the corresponding static inverse problem are
\begin{align}
        \gamma[u](x,x')
    &=\Tr[\Gamma[u]\hat\psi^\dagger(x')\hat\psi(x)],\\
    \mathcal U_\beta[\gamma]
    &:=\{u:\ -\gamma\in\partial W_\beta[u]\}\nonumber\\
    &=\{u:\ -u\in\partial F_\beta[\gamma]\}.
    \label{eq:1rdm_supporting_sources}
\end{align} 
This is precisely the static duality of
Eq.~\eqref{eq:static_LF_duality}, with a larger source and moment space
than for DFT.  The physical member satisfies
\begin{equation}
    \gamma_{\rm phys}:=\gamma[u_{\rm phys}],\qquad
    u_{\rm phys}\in\mathcal U_\beta[\gamma_{\rm phys}].
\end{equation}

In a finite orbital basis at positive temperature, 1RDMs with strictly
fractional occupations determine the nonlocal one-body source
uniquely~\cite{GiesbertzRuggenthaler2019}. With particle number fixed,
the source is instead unique up to an additive
constant~\cite{SutterGiesbertz2023}.
At zero temperature, nonuniqueness beyond gauge freedom can
arise~\cite{Giesbertz2015Invertibility}.

The density is the diagonal contraction
\begin{equation}
    \rho[\gamma](\mathbf r)
    =
    \sum_\sigma\gamma(\mathbf r,\sigma;\mathbf r,\sigma).
\end{equation}
Unlike DFT, the 1RDMFT inverse problem allows nonlocal one-body sources.
In a finite orbital basis, representability by a density operator on Fock
space is equivalent to \(0\leq\gamma\leq I\). Restricting the state to a fixed
\(N\)-particle sector additionally requires \(\Tr[\gamma]=N\), with further
constraints for pure states~\cite{Coleman1963,AltunbulakKlyachko2008}.
The intrinsic functional and its ground-state limit are
\begin{align}
    F_\beta[\gamma]
    &=\inf_{\Gamma\to\gamma}
    \left\{\Tr[\Gamma(T+W)]-\beta^{-1}S[\Gamma]\right\},\\
    F_\beta[\gamma]
    &\xrightarrow[\beta\to\infty]{}
    F[\gamma]
    =\inf_{\Gamma\to\gamma}\Tr[\Gamma(T+W)].
    \label{eq:1rdm_zero_temperature_limit}
\end{align}
The functionals \(F_\beta[\gamma]\) and \(W_\beta[u]\)
are related by Legendre--Fenchel duality.

The same duality with \(W=0\) supplies a noninteracting reference.
In a finite orbital basis at positive temperature, any 1RDM with strictly
fractional occupations is reproduced by
the unique one-particle generator
\(h_{\rm s}=\beta^{-1}\log(\gamma^{-1}-I)\)
\cite{GiesbertzRuggenthaler2019}.
The source \(u_{\rm s}\) is obtained by subtracting the fixed kinetic-energy
matrix from \(h_{\rm s}\).
At zero temperature, a nondegenerate noninteracting ground state
has an idempotent 1RDM, which prevents a direct Kohn--Sham-type
construction for interacting 1RDMs with fractional occupations.

For task evaluation, consider a one-body operator
\begin{equation}
    \hat A=\int dx\,dx'\,b(x,x')
    \hat\psi^\dagger(x)\hat\psi(x'),
\end{equation}
whose expectation value
\begin{equation}
    \langle\hat A\rangle
    =
    \int dx\,dx'\,b(x,x')\gamma(x',x)
\end{equation}
is obtained directly from \(\gamma\).  Two-body quantities and pair
correlations require a supporting one-body source or an additional 1RDM
functional.

\subsection{Imaginary-time enlargement: source-dependent density}
\label{sec:imaginary_time_density}

At finite temperature, the density-level description is enlarged by an
auxiliary local source \(j(\mathbf r,\tau)\) at a fixed static potential
\cite{Fukuda1994,ValievFernando1997,PolonyiSailer2002}.
The reference generator remains the fully interacting
\(K_{\rm phys}\), implicit in \(W_\beta[j]\).  Define
\begin{equation}
    K[j](\tau)
    =
    K_{\rm phys}
    +\int d\mathbf r\,
    j(\mathbf r,\tau)\hat n(\mathbf r).
    \label{eq:imaginary_time_density_generator}
\end{equation}
The physical system corresponds to \(j=0\). The source relations are
\begin{equation}
    \boxed{
    j\rightleftarrows\rho(\tau),
    \qquad
    j\longrightarrow\Omega .
    }
    \label{eq:imaginary_time_density_arrows}
\end{equation}
With the normalized imaginary-time pairing of Sec.~\ref{sec:framework},
the defining first variation is
\begin{equation}
    \delta W_\beta[j]
    =-\frac{1}{\beta}\int_0^\beta d\tau\int d\mathbf r\,
      \rho[j](\mathbf r,\tau)\,\delta j(\mathbf r,\tau).
    \label{eq:time_dependent_density}
\end{equation}
Here \(\rho[j](\mathbf r,\tau)\) is the normalized time-ordered density insertion
generated by \(j\) in
Eq.~\eqref{eq:imaginary_time_generating_functional}, reducing to an
ordinary Gibbs expectation when \(j\) is static.
The reverse arrow returns the imaginary-time-dependent local sources that
generate a given source-representable density,
\begin{equation}
    \Dec_\tau[\rho]
    :=\left\{j\mid\rho[j]=\rho\right\},
    \label{eq:time_dependent_density_decoder}
\end{equation}
with the source class and gauge convention understood.  On a smooth,
locally invertible source branch this relation admits the local Legendre
expression of Eq.~\eqref{eq:local_Legendre_expression}.  At the physical source,
the density is independent of imaginary time and reduces to the static
equilibrium density,
\begin{equation}
    \rho[0](\mathbf r,\tau)
    =
    \rho_{\rm phys}(\mathbf r).
    \label{eq:dynamic_to_static_density}
\end{equation}
Restricting \(j\) to time-independent real local sources gives
\(K[j](\tau)=K[v_{\rm phys}+j]\).
This slice therefore recovers the DFT source family of
Sec.~\ref{sec:dft} and its global Gibbs convexity.
The density \(\rho_{\rm phys}\) thus determines \(v_{\rm phys}\)
through the static inverse relation and hence the physical
system probed by \(j\).

Although the physical density is static, the source
dependence of \(W_\beta[j]\) encodes connected
imaginary-time density correlations through its
higher derivatives, as described in
Eq.~\eqref{eq:moment_cumulant_generators}.

Starting from the static Kohn--Sham potential
\(v_{\rm s}\) reproducing \(\rho_{\rm phys}\)
[Eq.~\eqref{eq:ks_potential}], introduce
a local imaginary-time auxiliary source \(j_{\rm s}\)
for the noninteracting reference:
\[
 K_{\rm s}[j_{\rm s}](\tau)
 =T+V[v_{\rm s}]
 +\int d\mathbf r\,
 j_{\rm s}(\mathbf r,\tau)\hat n(\mathbf r).
\]
On compatible local source--density branches,
the sources satisfy the density-matching relation
\[
 \rho_{\rm s}[j_{\rm s}]=\rho[j],
 \qquad
 \rho_{\rm s}[0]=\rho[0]=\rho_{\rm phys}.
\]
Let \(\Gamma_{\mathrm s,\beta}^{\rm loc}[\rho]\)
denote the corresponding noninteracting local
Legendre functional. The Legendre relations give
\begin{equation}
 \frac{\delta}{\delta\rho}
 \left(
 \Gamma_\beta^{\rm loc}[\rho]
 -\Gamma_{\mathrm s,\beta}^{\rm loc}[\rho]
 \right)
 =j_{\rm s}[\rho]-j[\rho].
 \label{eq:td_density_source_difference}
\end{equation}
The second-order structure of this representation is discussed
in the companion paper~\cite{ShengRPAHessian}.

The same construction extends to the 1RDM by allowing
spatially nonlocal, time-local one-body sources.
Real-time response and its invertibility have been studied in
time-dependent 1RDMFT~\cite{Pernal2007,Giesbertz2010,
Giesbertz2015Invertibility,Giesbertz2016Mixed}.

\subsection{Combined enlargement: finite-temperature Green's-function theory}
\label{sec:green_function}

The one-particle Matsubara Green's-function representation
uses a bilocal auxiliary source \(\eta(1,2)\), where
\(1=(x_1,\tau_1)\) and \(2=(x_2,\tau_2)\),
coupled to a fermion bilinear. Its moment is the
two-time Green's function \(G\)
\cite{Matsubara1955,MartinSchwinger1959}:
\begin{equation}
    \boxed{
    \eta\rightleftarrows G,
    \qquad
    \eta\longrightarrow\Omega .
    }
    \label{eq:gf_arrows}
\end{equation}
The physical source is \(\eta=0\).  Write
\(\hat\psi_{\rm phys}(1)=e^{\tau_1 K_{\rm phys}}\hat\psi(x_1)
e^{-\tau_1 K_{\rm phys}}\).
In the interaction picture, the time-local construction in
Eq.~\eqref{eq:imaginary_time_generating_functional} is also a Gibbs trace
with time-ordered source insertions.  The bilocal source extends this form
to two independent imaginary-time arguments.  Define
\begin{align}
    Z_\beta[\eta]
    &=\Tr\!\Biggl[
       e^{-\beta K_{\rm phys}}T_\tau\exp\!\Biggl(
       -\int d1\,d2\,\eta(2,1)\nonumber\\
    &\qquad\qquad\times\hat\psi_{\rm phys}^\dagger(2)\hat\psi_{\rm phys}(1)
       \Biggr)\Biggr],\nonumber\\
    W_\beta[\eta]&=\beta^{-1}\log Z_\beta[\eta].
    \label{eq:gf_source_partition}
\end{align}
Here \(T_\tau\) orders all fermionic fields, including those within each
bilocal insertion; equal-time products use the usual one-sided convention.
The bilocal source \(\eta(2,1)\) couples fields at independent
imaginary times directly in the time-ordered trace generated
by \(K_{\rm phys}\).
Equation~\eqref{eq:gf_source_partition} nonetheless defines an operator
generating functional built from the fixed physical generator \(K_{\rm phys}\).
With the bilocal pairing
\(\langle\eta,G\rangle_\beta=\beta^{-1}\int d1\,d2\,
\eta(2,1)G(1,2)\), the first variation of
\(W_\beta[\eta]\) reads
\begin{align}
    \delta W_\beta[\eta]
    &=-\langle\delta\eta,G[\eta]\rangle_\beta,\nonumber\\
    G[\eta](1,2)
    &=-\left\langle T_\tau
    \hat\psi_{\rm phys}(1)\hat\psi_{\rm phys}^\dagger(2)\right\rangle_\eta .
    \label{eq:gf_forward}
\end{align}
The brackets denote the normalized, source-weighted time-ordered insertion;
at \(\eta=0\) they reduce to the ordinary Hamiltonian equilibrium trace.

The normalized partition function \(Z_\beta[\eta]/Z_\beta[0]\)
and its logarithm generate full time-ordered correlations and
cumulants of the bilinear insertions, respectively, through unnormalized source-kernel derivatives at \(\eta=0\),
with a factor \((-1)^n\) for \(n\) derivatives.
We use the bilocal generating functional where \(Z_\beta[\eta]\) is finite
and nonzero and admits a consistent local logarithm branch.

For a general bilocal source \(\eta\), the Green's function need not be
invariant under imaginary-time translations. At the physical source
\(\eta=0\), however, \(G_{\rm phys}:=G[0]\) depends only on the
time difference. Suppressing spatial and spin labels,
\[
    G_{\rm phys}(\tau_1,\tau_2)=G_{\rm phys}(\tau_1-\tau_2).
\]
The physical Green's function admits the fermionic Matsubara representation,
\begin{equation}
\begin{aligned}
    G_{\rm phys}(\tau_1-\tau_2)
    &\longleftrightarrow
    G_{\rm phys}(i\omega_n),\\
    \omega_n&=\frac{(2n+1)\pi}{\beta}.
\end{aligned}
    \label{eq:gf_matsubara}
\end{equation}
The physical Green's function satisfies the usual antiperiodicity,
analyticity, spectral, and equal-time conditions.

The reverse arrow is the inverse-kernel set
\begin{equation}
    \mathcal B_\beta[G]
    :=\left\{\eta\mid G[\eta]=G\right\}.
    \label{eq:gf_supporting_sources}
\end{equation}
Source representability requires
\(\mathcal B_\beta[G]\neq\varnothing\).
On a locally invertible source branch, the Legendre construction of
Eq.~\eqref{eq:local_Legendre_expression} defines a functional of \(G\)
whose first variation is \(-\eta[G]\).
The physical Green's function \(G_{\rm phys}\) is a stationary
point, with \(\eta[G_{\rm phys}]=0\).
Its equal-time limit gives \(\gamma_{\rm phys}\), which
determines \(u_{\rm phys}\) through the static inverse relation
and hence the physical system probed by \(\eta\).

We construct a noninteracting reference by removing the interaction
from \(K_{\rm phys}\) and introducing a bilocal source \(\eta_{\rm s}\).
Denoting the bare Green's function by \(G_0\), we have
\(G_{\rm s}[\eta_{\rm s}]=(G_0^{-1}-\eta_{\rm s})^{-1}\).
On a common local domain of the interacting and noninteracting
inverse maps around \(G_{\rm phys}\), the source reproducing
\(G\) is given by \(\eta_{\rm s}[G]=G_0^{-1}-G^{-1}\).

The difference between the interacting and noninteracting local
Legendre functionals on this domain defines the
Luttinger--Ward functional \(\Phi[G]\).
Its first variation defines the self-energy through
\(\delta\Phi[G]=\langle\Sigma[G],\delta G\rangle_\beta\).
The Legendre relations then give
\(\Sigma[G]=\eta_{\rm s}[G]-\eta[G]\), the difference between
the noninteracting and interacting sources reproducing the same \(G\).
At the interacting physical point \(G_{\rm phys}\), stationarity of the
interacting local Legendre functional (the Baym--Kadanoff functional)
gives \(\eta[G_{\rm phys}]=0\), hence
\(\Sigma[G_{\rm phys}]=\eta_{\rm s}[G_{\rm phys}]\).
The noninteracting reference relation then yields Dyson's equation
\(G_{\rm phys}^{-1}=G_0^{-1}-\Sigma[G_{\rm phys}]\)
\cite{LuttingerWard1960,BaymKadanoff1961,Baym1962,ChitraKotliar2001}.
The construction is local to a chosen inverse-source branch,
consistent with the possible multivaluedness of the
Luttinger--Ward functional~\cite{KozikFerreroGeorges2015}.

In diagrammatic perturbation theory, \(\Phi[G]\) admits a
formal expansion in two-particle-irreducible (2PI) vacuum
skeleton diagrams built from the full \(G\).
Its functional derivative generates the corresponding
one-particle-irreducible (1PI) self-energy skeleton
diagrams~\cite{LuttingerWard1960,CornwallJackiwTomboulis1974}.
% Approximating \(\Phi[G]\) and solving the resulting Dyson equation
% self-consistently yields \(\Phi\)-derivable approximations, such as
% self-consistent \(GW\)~\cite{Hedin1965,Almbladh1999}.
% On regular stationary branches, source differentiation of the
% stationary value generates the corresponding approximate \(G\).
For an equilibrium system with instantaneous two-body interactions,
the Galitskii--Migdal identity~\cite{GalitskiiMigdal1958} gives the
total electronic energy in terms of \(G_{\rm phys}\) and the one-body
part of the physical Hamiltonian.
More general two- and higher-particle correlation functions can be
obtained through higher-order source derivatives.

\section{Information enlargement and source-class refinement}
\label{sec:refinement}

\subsection{Moment reduction and higher-order extensions}
\label{subsec:moment_reduction}

The four representations are related through moment contractions
and source restrictions. Consider two representations \(X\) and
\(Y\) related by a bounded linear contraction \(m_Y=Pm_X\).
The corresponding source embedding \(P^*\) is defined by
\begin{equation}
    \langle P^*J,m_X\rangle_X
    =\langle J,Pm_X\rangle_Y.
    \label{eq:source_embedding}
\end{equation}
For compatible source families describing the same physical
system with aligned source origins, we have
\begin{equation}
    W_Y[J]=W_X[P^*J].
    \label{eq:compatible_sources}
\end{equation}
Taking the first derivative gives
\begin{equation}
    m_Y[J]=Pm_X[P^*J].
    \label{eq:moment_projection}
\end{equation}
Expressing both decoders in terms of the common source
parameter \(J\), we have
\begin{equation}
 \{J:m_X[P^*J]=m_X\}
 \subseteq
 \{J:m_Y[J]=Pm_X\}.
 \label{eq:decoder_inclusion}
\end{equation}

Choosing a complement to \(\ker P\) (with a bounded splitting
in infinite dimensions), we write
\begin{equation}
    \operatorname{dom}P\simeq
    \operatorname{ran}P\oplus\ker P.
    \label{eq:moment_splitting}
\end{equation}
In these coordinates, \(m_X=(q,y)\), with \(q=Pm_X\)
and \(y\in\ker P\), and
\begin{equation}
    P(q,y)=q,\qquad P^*J=(J,0).
    \label{eq:adapted_source_embedding}
\end{equation}

Let \(\Gamma_X[q,y]\) and \(\Gamma_Y[q]\) be the moment
functionals in the aligned source coordinates. 
On a regular branch generated by the embedded sources,
\begin{equation}
    \frac{\delta\Gamma_Y}{\delta q}
    =\frac{\delta\Gamma_X}{\delta q}=-J,
    \qquad
    \frac{\delta\Gamma_X}{\delta y}=0.
    \label{eq:elimination_stationarity}
\end{equation}
The Legendre relations then give
\begin{equation}
    \Gamma_Y[q]=\Gamma_X[q,y[q]]
    =\underset{y}{\operatorname{stat}}\,\Gamma_X[q,y],
    \label{eq:stationary_moment_reduction}
\end{equation}
where \(\operatorname{stat}\) denotes the stationary value on the
chosen branch and \(y[q]\) is the corresponding solution.
For static representations with the same state class and
source-independent generator, the constrained search gives
\begin{equation}
    F_Y[q]=\inf_{m_X:\,Pm_X=q}F_X[m_X]
          =\inf_y F_X[q,y].
    \label{eq:static_moment_reduction}
\end{equation}
For static representations, the local Legendre functional
coincides with the constrained-search functional on the
inverse-source branch when the same source coordinates are used.
Shifting the source origin to \(s_{\rm phys}\) through
\(J=s-s_{\rm phys}\) gives
\begin{equation}
    \Gamma[m]=F_\beta[m]+\langle s_{\rm phys},m\rangle.
\end{equation}
The second-order structure of these reductions is discussed in
the companion paper~\cite{ShengRPAHessian}.

We now give the contractions and source embeddings for the four
representations.
For the 1RDM and static density, we have
\begin{equation}
\begin{aligned}
    \rho[\gamma](\mathbf r)
    &=\sum_\sigma\gamma(\mathbf r\sigma,\mathbf r\sigma),\\
    u[v](x,x')&=v(\mathbf r)\delta(x-x').
\end{aligned}
    \label{eq:gamma_density_contraction_embedding}
\end{equation}
The corresponding constrained search is
\begin{equation}
    F_\beta[\rho]
    =\inf_{\gamma:\,\rho[\gamma]=\rho}F_\beta[\gamma].
    \label{eq:gamma_density_functional_reduction}
\end{equation}
Time averaging reduces the imaginary-time density to the static
density, with a time-independent auxiliary source:
\begin{equation}
\begin{aligned}
    \rho[\rho(\cdot,\cdot)](\mathbf r)
    &:=\frac1\beta\int_0^\beta\rho(\mathbf r,\tau)\,d\tau,\\
    j[v-v_{\rm phys}](\mathbf r,\tau)
    &:=v(\mathbf r)-v_{\rm phys}(\mathbf r).
\end{aligned}
    \label{eq:time_density_contraction_embedding}
\end{equation}
Writing \(\rho(\mathbf r,\tau)=q(\mathbf r)+y(\mathbf r,\tau)\),
with \(y\) of zero time average, static sources give \(y=0\).
On this branch,
\begin{equation}
    \Gamma_\beta^{\rm loc}[q,0]
    =F_\beta[q]+\langle v_{\rm phys},q\rangle.
    \label{eq:time_density_functional_reduction}
\end{equation}
The imaginary-time density is the spatial diagonal of the
one-sided equal-time Green's function:
\begin{equation}
\begin{aligned}
    \rho[G](\mathbf r,\tau)
    &:=\sum_\sigma
    G(\mathbf r\sigma,\tau;\mathbf r\sigma,\tau^+),\\
    \eta[j](2,1)&=j(\mathbf r_1,\tau_1)\delta(x_2-x_1)\\
    &\qquad\times\delta(\tau_2-\tau_1).
\end{aligned}
    \label{eq:gf_equal_time_density}
\end{equation}
In the bilocal pairing, the time delta function selects
the one-sided limit \(\tau_2=\tau_1^+\).
On the corresponding local branch,
\begin{equation}
    \Gamma_\beta^{\rm loc}[\rho]
    =\underset{G:\,\rho[G]=\rho}{\operatorname{stat}}
      \Gamma_\beta^{\rm loc}[G].
    \label{eq:gf_time_density_functional_reduction}
\end{equation}
The static 1RDM is obtained by averaging the equal-time
Green's function over imaginary time:
\begin{equation}
\begin{aligned}
    \gamma[G](x,x')
    &:=\frac1\beta\int_0^\beta
    G(x,\tau;x',\tau^+)\,d\tau,\\
    \eta[u-u_{\rm phys}](2,1)
    &:=(u-u_{\rm phys})(x_2,x_1)\\
    &\qquad\times\delta(\tau_2-\tau_1).
\end{aligned}
    \label{eq:gf_static_1rdm}
\end{equation}
The corresponding functional reduction is
\begin{equation}
    F_\beta[\gamma]+\langle u_{\rm phys},\gamma\rangle
    =\underset{G:\,\gamma[G]=\gamma}{\operatorname{stat}}
      \Gamma_\beta^{\rm loc}[G].
    \label{eq:gf_1rdm_functional_reduction}
\end{equation}
A related reduction is developed in Ref.~\cite{Blochl2013}.
At equilibrium, time-translation invariance gives
\begin{equation}
    \gamma_{\rm phys}(x,x')
    =G_{\rm phys}(x,\tau;x',\tau^+).
    \label{eq:gf_equal_time_1rdm}
\end{equation}
The two contraction paths from \(G\) to the static density agree,
\[
    \rho[\gamma[G]](\mathbf r)
    =\frac1\beta\int_0^\beta
    \rho[G](\mathbf r,\tau)\,d\tau,
\]
and give the functional reduction
\begin{equation}
    F_\beta[\rho]+\langle v_{\rm phys},\rho\rangle =\underset{G:\,\rho[\gamma[G]]=\rho}{\operatorname{stat}}
      \Gamma_\beta^{\rm loc}[G].
    \label{eq:gf_static_density_functional_reduction}
\end{equation}

The physical Green's function \(G_{\rm phys}=G[0]\)
also permits explicit recovery of the underlying static
one-body source.
For fixed \(T\) and \(W\), the coefficient \(M_1\)
of \((i\omega_n)^{-2}\) in the high-frequency expansion
of \(G_{\rm phys}(i\omega_n)\) satisfies
\(M_1=h+\Sigma_{\rm HF}[\gamma_{\rm phys}]\),
where \(h=t+u_{\rm phys}\), \(t\) is the kinetic
matrix, and \(\Sigma_{\rm HF}[\gamma_{\rm phys}]\) is the
Hartree--Fock self-energy constructed from the exact
1RDM of Eq.~\eqref{eq:gf_equal_time_1rdm}
\cite{WangDangMillis2011}.
Thus \(u_{\rm phys}=M_1-
\Sigma_{\rm HF}[\gamma_{\rm phys}]-t\).
Unlike the generally implicit density--potential and
1RDM--potential inversions on their respective static
source families, the full physical \(G_{\rm phys}\) provides
an explicit recovery of the one-body source.
This relation also holds in the zero-temperature limit
when the relevant moments exist, even where the static
inverse maps need not be unique.

The source constructions also extend to higher particle rank.
Coupling a source directly to \(k\)-body operators makes the
\(k\)-RDM a first-order reduced moment, which determines
expectation values of \(k\)-body observables and obeys the
corresponding representability constraints, including
\(N\)-representability when particle number is fixed
\cite{Lowdin1955,Coleman1963,Mazziotti2012,WangLiebertPenzSchilling2026}.
Similarly, a \(k\)-particle Matsubara Green's function retains
\(2k\) fermionic fields at independent imaginary times
\cite{MartinSchwinger1959} and becomes a first-order reduced
moment when a multilocal source is coupled directly to the
corresponding field product
\cite{CornwallJackiwTomboulis1974,Berges2004}.
Its appropriate equal-time limit recovers the \(k\)-RDM
with consistent fermionic operator ordering.

Higher-order correlations can also be generated by repeated
source derivatives. Linear Grassmann field sources generate
correlations of individual fermionic fields
\cite{NegeleOrland1988}, while the bilocal source \(\eta\)
generates correlations of fermionic bilinears.
Normalized derivatives of the partition function give full
correlations, while derivatives of its logarithm give cumulants
of the corresponding insertions.
Higher-rank RDMs can also be recovered from correlations generated
by imaginary-time-dependent nonlocal one-body sources through
appropriate equal-time limits and operator reordering.
Static one-body source derivatives instead yield
imaginary-time-integrated responses.

\subsection{Joint moment representations and interaction sources}
\label{subsec:interaction_sources}

In the Green's-function representation of
Sec.~\ref{sec:green_function}, the bilocal source
\(\eta\) is varied while the interaction is held fixed.
Almbladh et al.~\cite{Almbladh1999} introduced a functional
of both the Green's function and the screened interaction.
To include the interaction in the source family, we introduce
a bilocal interaction source \(U(\mathbf r,\tau;\mathbf r',\tau')\).
Adding
\[
 -\frac12\int d1\,d2\,U(1,2)\hat n(1)\hat n(2)
\]
to the time-ordered exponential in
Eq.~\eqref{eq:gf_source_partition} defines
\(W_\beta[\eta,U]\), with first variation
\begin{equation}
 \delta W_\beta[\eta,U]
 =-\langle\delta\eta,G\rangle_\beta
 -\frac12\langle\delta U,Q\rangle_\beta,
\end{equation}
where
\[
 \begin{aligned}
 Q(1,2)
 &=\langle T_\tau\hat n(1)\hat n(2)\rangle_{\eta,U},\\
 \langle U,Q\rangle_\beta
 &=\frac1\beta\int d1\,d2\,U(1,2)Q(1,2).
 \end{aligned}
\]
Here \(1=(\mathbf r_1,\tau_1)\) for density kernels,
and the operators are evolved with \(K_{\rm phys}\).
The enlarged source--moment relation is
\[
 \boxed{(\eta,U)\rightleftarrows(G,Q).}
\]
Following Eq.~\eqref{eq:local_Legendre_expression},
the joint local Legendre functional is
\begin{equation}
 \Gamma_\beta^{\rm loc}[G,Q]
 =-W_\beta[\eta,U]-\langle\eta,G\rangle_\beta
 -\frac12\langle U,Q\rangle_\beta,
 \label{eq:joint_legendre}
\end{equation}
with
\[
 \delta\Gamma_\beta^{\rm loc}[G,Q]
 =-\langle\eta,\delta G\rangle_\beta
 -\frac12\langle U,\delta Q\rangle_\beta.
\]

Performing the Legendre transform with respect to
\(\eta\) at fixed \(U\) gives
\[
 \Gamma_\beta^{\rm loc}[G;U]
 =-W_\beta[\eta,U]-\langle\eta,G\rangle_\beta,
\]
with first variation
\[
 \delta\Gamma_\beta^{\rm loc}[G;U]
 =-\langle\eta,\delta G\rangle_\beta
 +\frac12\langle\delta U,Q\rangle_\beta.
\]
As in the Luttinger--Ward construction of
Sec.~\ref{sec:green_function}, define
\[
 \Phi[G;U]
 =\Gamma_\beta^{\rm loc}[G;U]
 -\Gamma_{\mathrm s,\beta}^{\rm loc}[G].
\]
The Legendre relations give
\begin{equation}
 \delta\Phi[G;U]
 =\langle\Sigma,\delta G\rangle_\beta
 +\frac12\langle\delta U,Q\rangle_\beta,
\end{equation}
where
\(\Sigma[G;U]=\eta_{\mathrm s}[G]-\eta[G;U]\).
Separating the Hartree contribution,
\begin{equation}
    \begin{aligned}
        \Phi[G;U]&=\Phi_{\mathrm H}[G;U]+\Phi_{\mathrm{xc}}[G;U],\\
 \Phi_{\mathrm H}[G;U]&
 =\frac12\langle(v_c+U)\rho,\rho\rangle_\beta,
    \end{aligned}
\end{equation}
with \(\rho=\rho[G]\), gives
\begin{equation}
 \delta\Phi_{\mathrm{xc}}[G;U]
 ={}\langle\Sigma_{\mathrm{xc}},\delta G\rangle_\beta+\frac12\langle\delta U,Q-\rho\otimes\rho\rangle_\beta,
 \label{eq:interaction_derivative_variation}
\end{equation}
where \(v_{\rm c}\) denotes the instantaneous Coulomb kernel, \(\Sigma_{\mathrm H}:=
\delta\Phi_{\mathrm H}[G;U]/\delta G\)
and \(\Sigma_{\mathrm{xc}}:=\Sigma-\Sigma_{\mathrm H}\).

Define
\begin{equation}
 \begin{aligned}
 \chi(1,2)
 &:=-Q(1,2)+\rho(1)\rho(2)\\
 &=-\langle T_\tau\delta\hat n(1)
                    \delta\hat n(2)\rangle_{\eta,U},
 \end{aligned}
 \label{eq:connected_density_correlation}
\end{equation}
where \(\delta\hat n=\hat n-\rho\).
Subtracting the Hartree contribution gives
\begin{equation}
    \frac{\delta\Phi_{\mathrm{xc}}[G;U]}{\delta G}
    =\Sigma_{\mathrm{xc}},
    \qquad
    \frac{\delta\Phi_{\mathrm{xc}}[G;U]}{\delta U}
    =-\frac12\chi.
    \label{eq:interaction_derivative}
\end{equation}
The second-order variational structure associated
with \(\chi\) is discussed in the companion
paper~\cite{ShengRPAHessian}.

\section{Approximations and quantum embedding}
\label{sec:embedding}

\subsection{Static variational approximations}

We can construct variational approximations by introducing
either a source-independent trial-state class \(\mathcal A\)
or a source-independent approximate intrinsic functional
\(F_\beta^{\rm app}[m]\).
The corresponding source functionals are
\begin{equation}
\begin{aligned}
 W_\beta^{\mathcal A}[s]
 &=\sup_{\Gamma\in\mathcal A}\Bigl\{-\Tr[\Gamma K_0]
   -\langle s,m[\Gamma]\rangle\\
 &\hspace{29mm}+\beta^{-1}S[\Gamma]\Bigr\},\\
 W_\beta^{\rm app}[s]
 &=\sup_m\bigl\{-F_\beta^{\rm app}[m]-\langle s,m\rangle\bigr\}.
\end{aligned}
 \label{eq:approx_variational_generators}
\end{equation}
Both remain
suprema of affine functions of \(s\), so they are convex; at attained,
differentiable optima, their source derivatives generate the approximate
moments. The two approximation strategies are related: restricting the
state search to \(\mathcal A\) induces an approximate intrinsic
functional through constrained search at fixed \(m\), as in
Levy's ground-state construction~\cite{Levy1979}.
This functional reproduces \(W_\beta^{\mathcal A}[s]\)
through the second expression of
Eq.~\eqref{eq:approx_variational_generators},
whereas a directly approximated intrinsic functional need not
arise from a restricted trial-state class.
For a nonconvex trial class, the induced intrinsic functional need not
be convex; under the corresponding duality assumptions, the reverse
Legendre--Fenchel transform recovers its lower-semicontinuous convex
envelope~\cite{Lieb1983}.

At the ground-state level, Hartree--Fock uses the
source-independent class of all admissible Slater determinants.
More generally, variational wavefunction methods can be
formulated with source-independent trial classes by fixing
the one-particle basis and ansatz restrictions.
Truncated configuration interaction with a fixed orbital basis and
a fixed selection of determinants provides a concrete example
\cite{SherrillSchaefer1999}.
In practice, however, truncated wavefunction methods often
use orbitals obtained from Hartree--Fock calculations.
If these orbitals are recomputed at each \(s\), the trial
class may become source-dependent, as in truncated CI with
\(\mathcal A[s]\).
The resulting CI energy remains variational at each source,
but the preceding convexity argument no longer applies,
and source derivatives can acquire orbital-response contributions
beyond the CI expectation values~\cite{HandySchaefer1984}.

Practical Kohn--Sham DFT illustrates the functional strategy by
retaining \(F_{{\rm s},\beta}\) and \(E_{\rm H}\) introduced in
Sec.~\ref{sec:dft} while approximating
\(F_{{\rm xc},\beta}\). At zero temperature, the local-density
approximation~\cite{KohnSham1965} uses the uniform electron gas
at the local density; generalized-gradient approximations such as
PBE~\cite{PerdewBurkeErnzerhof1996} also use density gradients.
Minimizing the resulting intrinsic functional together with the
external-potential coupling yields the approximate density and
energy, or free energy at finite temperature.
In ground-state 1RDMFT, the kinetic energy is explicit
in the 1RDM, while the interaction functional is
approximated, for example, by the Müller or power
functionals expressed through natural orbitals and
occupation numbers~\cite{Muller1984,Sharma2008}.

\subsection{Stationary auxiliary-source approximations}

For auxiliary sources, one may approximate the local
effective action \(\Gamma_\beta^{\rm loc}[m]\) of
Eq.~\eqref{eq:local_Legendre_expression}.
The approximate moment is determined by stationarity of
\(\Gamma_\beta^{\rm app}[m]+\langle j,m\rangle_\beta\),
reducing to stationarity of \(\Gamma_\beta^{\rm app}[m]\)
at the physical source \(j=0\).
On a regular branch, this construction preserves the
source--moment derivative relation.
In the Green's-function representation of
Sec.~\ref{sec:green_function}, approximating \(\Phi[G]\)
gives a corresponding self-energy approximation through
functional differentiation. Solving the resulting Dyson
equation self-consistently yields a \(\Phi\)-derivable
approximation~\cite{Baym1962}.
Retaining the Hartree and exchange contributions to
\(\Phi[G]\) gives Hartree--Fock, while including the RPA
correlation functional evaluated with the full \(G\)
yields self-consistent \(GW\)
\cite{Hedin1965,Almbladh1999,ShengRPAHessian}.
Hartree--Fock thus illustrates both a restricted-state search
and a functional approximation.

Approximate functionals in one representation can also be
constructed from those in another through functional reduction.
An example is the construction of approximate 1RDM functionals
from the Luttinger--Ward functional~\cite{Blochl2013}.
Compatibility under reduction provides a guiding principle
for constructing approximations across representations,
as explored further in the companion paper~\cite{ShengRPAHessian}.

\subsection{Quantum embedding as an application}

Quantum embedding constructs an effective local Hamiltonian or
action that retains environmental effects, allowing selected
degrees of freedom to be treated by a correlated solver.
Global and local descriptions can be coupled through
consistency conditions on selected interface moments.
The following examples illustrate density-, 1RDM-, and
Green's-function-based interfaces, corresponding to the
representations of Sec.~\ref{sec:applications}.

In density-functional embedding theory
(DFET)~\cite{HuangPavoneCarter2011,YuCarter2017}, a common embedding potential
\(v_{\rm emb}\) is optimized so that the densities of two
embedded DFT subsystems reproduce the global DFT density:
\begin{equation}
    \rho_A^{\rm DFT}[v_{\rm emb}](\mathbf r)
    =
    \rho_{\rm g}^{\rm DFT}(\mathbf r)
    -
    \rho_B^{\rm DFT}[v_{\rm emb}](\mathbf r).
\end{equation}
The embedding potential acts as a Lagrange multiplier
enforcing this density-partition constraint.

In self-consistent DMET~\cite{KniziaChan2012,Wouters2016DMET},
a correlation potential \(u\) modifies an auxiliary global
one-body problem whose Slater determinant defines bath orbitals.
The potential is adjusted to match selected fragment 1RDM
elements between the auxiliary global and correlated
fragment-plus-bath calculations.
For fragment-block matching at a self-consistent potential \(u_*\),
\begin{equation}
    \gamma_{{\rm g},{\rm frag}}[u_*]
    \simeq\gamma_{{\rm l},{\rm frag}}[u_*].
\end{equation}
Updating \(u\) changes the bath and hence the correlated
embedding problem. Exact matching may be unattainable within the restricted
source family~\cite{Faulstich2022DMET}.

DMFT~\cite{Georges1996} uses the impurity self-energy in the
lattice Dyson equation and updates the Weiss field or hybridization
so that the impurity Green's function equals the projected lattice
Green's function:
\begin{equation}
    G_{\rm imp}(i\omega_n)
    =P_{\rm loc}G_{\rm lat}(i\omega_n)P_{\rm loc}^\dagger .
\end{equation}
The hybridization \(\Delta\) is a bilocal one-body source
restricted to causal, bath-realizable kernels.
These consistency conditions constrain the selected interface moments.

In DFET, density matching is enforced at the DFT level,
while the resulting embedding potential is held fixed
in subsequent correlated wavefunction calculations
\cite{Libisch2014}.
QDET~\cite{ShengQDET2022} constructs a one-shot effective
Hamiltonian through \(G_0W_0\)-based downfolding, with
a double-counting correction derived at the same level
of theory.
Unlike DMET and DMFT, neither construction feeds the
correlated results back into global--local moment matching.

\section{Conclusion}
\label{sec:conclusion}

The source--moment relations \(J\rightleftarrows m\) and \(J\to\Omega\)
provide a common description of equilibrium reduced-moment representations:
a source generates the retained object, its inverse relation identifies
compatible sources, and a reader evaluates other quantities.  Static Hermitian
source families admit Gibbs convex duality.  Auxiliary sources about a fixed
interacting physical generator provide a broader setting for equilibrium
moments and correlations.

The density, 1RDM, source-dependent imaginary-time density, and one-particle
Green's function retain different spatial and temporal information.  On the
common physical local-potential family, the density is the least detailed
of these one-body moments under source-independent contraction.  Higher-rank
reduced moments fit the same organization.
Comparing interacting and noninteracting source--moment relations clarifies the
Kohn--Sham potential relation, finite-temperature 1RDM references, the imaginary-time
exchange--correlation kernel, and the Dyson self-energy.  Variational and stationary
approximations can retain the source--moment
structure through a restricted state search or an approximate intrinsic
functional or effective action.
Quantum embedding uses selected interface moments to couple global
and local descriptions, as illustrated by DFET, DMET, and DMFT.

\bibliography{ref}

\end{document}